\documentclass[preprints,article,accept,moreauthors,pdftex]{Definitions/mdpi} 

\usepackage{algorithm}
\usepackage{algorithmicx}
\usepackage{algpseudocode}
\usepackage{setspace}

\newcommand{\mm}[1]{\begin{align} #1 \end{align}}

\newcommand{\ket}[1]{\ensuremath\,|#1\rangle}
\newcommand{\bra}[1]{\ensuremath\langle #1 |\,}
\newcommand{\braket}[2]{\ensuremath \langle #1 | #2 \rangle}
\newcommand{\ketbra}[2]{|#1\rangle \langle #2 |}

\newcommand{\tr}{\text{Tr}}
\renewcommand{\rho}{\varrho}

\newcommand{\ot}{\otimes}
\newcommand{\dg}{^\dagger}

\newcommand{\lb}{\left(}
\newcommand{\rb}{\right)}

\newcommand{\id}{I}

\newtheorem{lemma}[theorem]{Lemma}
\newtheorem{corollary}[theorem]{Corollary}

\firstpage{1} 
\pubvolume{xx}
\issuenum{1}
\articlenumber{5}
\pubyear{2019}
\copyrightyear{2019}
\history{}

\Title{Daemonic Ergotropy: Generalised Measurements and Multipartite Settings}

\Author{Fabian Bernards$^{1}$, Matthias Kleinmann $^{1}$, Otfried G\"uhne $^{1}$ and Mauro Paternostro $^{2,}$*}

\AuthorNames{Fabian Bernards, Matthias Kleinmann, Otfried G\"uhne and Mauro Paternostro}

\address{%
$^{1}$ \quad Naturwissenschaftlich-Technische Fakult\"at, Universit\"at Siegen, Walter-Flex-Stra\ss{e} 3, 57068 Siegen, Germany\\
$^{2}$ \quad School of Mathematics and Physics, Queen's University, Belfast BT7 1NN, UK}

\abstract{Recently, the concept of daemonic ergotropy has been
introduced to quantify the maximum energy that can be
obtained from a quantum system through an ancilla-assisted work extraction protocol based on information gain via
projective measurements [G. Francica
{\it et al.}, npj Quant. Inf. {\bf 3}, 12 (2018)]. 
We prove that quantum correlations are
not advantageous over classical correlations if projective measurements
are considered. We go beyond the limitations of the original definition
to include generalised measurements and provide an example in which
this allows for a higher daemonic ergotropy. Moreover, we propose a
see-saw algorithm to find a measurement that attains the maximum work
extraction. Finally, we provide a multipartite generalisation of
daemonic ergotropy that pinpoints the influence of multipartite quantum
correlations, and study it for multipartite entangled and classical~states.}


\begin{document}


\section{Introduction}
In the rapidly evolving research arena embodied by the thermodynamics of quantum
systems, 
the resource-role of quantum features in work-extraction protocols is one of the most interesting and pressing 
open questions~\cite{review,review2,review3,review4}. Quantum coherences are claimed to be
responsible for the extraction of work from a single heat
bath~\cite{scully} and the enhanced performance of quantum engines~\cite{karimi}. Weakly driven quantum heat engines are
known to exhibit enhanced power outputs with respect to their
classical (stochastic) versions~\cite{uzdin}. Quantum information-assisted
schemes for energy extraction have been put forward and shown to be potentially
able to achieve significant
efficiencies~\cite{jordan,elouard,yi,gelbwaser,jacobs,abah}. However,
controversies in the usefulness of quantum correlations and coherences in
schemes for the extraction of work from quantum systems have also been
discussed~\cite{acin,acing,fusco,michele}. While a full physical understanding
of these issues is still far from being acquired, theoretical progress in this direction will be key to the design and implementation of informed experimental proof-of-principle experiments and thus the consolidation of a quantum approach to the thermodynamics of microscopic systems. 

Recently, a simple ancilla-assisted work-extraction protocol has been proposed
that is able to pinpoint the crucial role that quantum measurements have in the
performance of a quantum work-extraction game. This protocol also highlighted important
implications arising from the availability of quantum correlations between the
work medium and the ancilla~\cite{daem}. The scheme provided a link between enhanced work extraction capabilities and quantum
entanglement between ancilla and work medium, suggesting the possibility to exploit
entanglement as a resource. 

In this work we show that although this link exists for pure states,
quantum correlations and work extraction capabilities are unrelated if mixed
states are considered. 
However, the scheme in Reference~\cite{daem} relied on a set of very stringent
assumptions, which leave room to further investigations aimed at clarifying the
potential benefits of exploiting quantum resources. 
Here, we critically investigate the protocol in Reference~\cite{daem}, and extend it in
various directions. First, we address the class of measurements that
ensure the enhancement of the work-extraction performance. We provide an example
in which generalised measurements allow for more extracted energy than
projective measurements do. The search for the right generalised measurement
poses serious computational challenges that we solve by proposing a constructive
see-saw algorithm that is able to identify the most effective
measurement for a given state of the work medium and ancilla, and an assigned
Hamiltonian of the former. We then address the issue embodied by the interplay
between information gathered via optimal measurements and quantum correlations
shared between work medium and ancilla. We show that, depending on the nature of
the optimal measurement, quantum correlations  may become entirely inessential for the enhancement of work
extraction. Finally, we open the investigation to multipartite settings by
addressing the case of multiple work media and ancillas, showing that the
structure of correlation-sharing among the various parties of such a
system is key in the performance of our work-extraction
protocol.

Our results contribute to the ongoing research for the ultimate
resources to be exploited to draw an effective and useful framework for quantum
enhanced thermodynamical processes. While clarifying a number of important
points, our work opens up new avenues of investigation that will be crucial for the design of unambiguous experimental validations.

 
 \section{Notation and Concepts}
 
 The maximal energy decrease of a given state $\rho^S$ with respect to a
reference Hamiltonian $H$ undergoing an arbitrary unitary evolution $U$ is its
{ergotropy} \cite{ergotropy}
\begin{align}
W(\rho^S, H) = \tr[\rho^S H] - \min_U \tr[U\rho^S U\dg H]. \label{eq-1}
\end{align}

This is interpreted as the maximal amount of work that can be extracted from a system prepared in state $\rho^S$
by the means of a unitary protocol~\cite{ergotropy}.
Given some
state in its spectral decomposition $\rho^S~=~\sum_k r_k~\ketbra{r_k}{r_k}$ with $r_{k+1} \le r_k$
and a Hamiltonian $H = \sum_k \epsilon_k
\ketbra{\epsilon_k}{\epsilon_k}$ with $\epsilon_{k+1} \ge \epsilon_k$ 
 the optimal
unitary is $U~=~\sum_k \ketbra{\epsilon_k}{r_k}$ \cite{ergotropy}. {This is a direct consequence of the von Neumann trace inequality~\cite{Mirsky1975}. It states that ${\rm tr}[AB]~\le~\sum_i a_i b_i$, where $a_i \, (b_i)$ are the eigenvalues of $A \, (B)$ and $a_{i+1} \ge
a_i,\quad b_{i+1}~\ge~b_i$. Choosing $A = -U\rho^S U\dagger$ and $B = H$ and writing $\max_U {\rm tr}[-U\rho^S U\dagger H] = - \min_U {\rm tr}[U \rho^S U\dagger H]$ then shows that the bound given by the von Neumann trace inequality is achieved with the unitary stated~above.}

In Reference~\cite{daem}, an ancilla-assisted protocol allowed for enhanced work extraction by making use of a process of information inference.  
The fundamental building blocks of the protocol are embodied by the joint state of a work medium $S$ and an ancilla $A$, and a projective
measurement $M$ performed on the latter (cf. Figure~\ref{scheme}). The
information gathered through these measurements is then used to determine a unitary transformation to be applied to $S$ to extract as much work as possible.

\begin{figure}[H]
\centering
\includegraphics[width=.25\columnwidth]{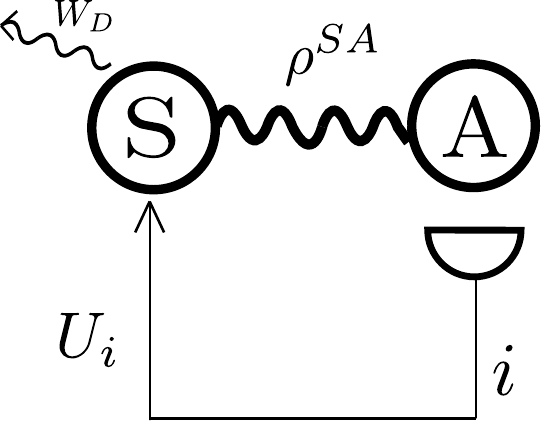}
\caption{Illustration of daemonic ergotropy. A system $S$ is coupled to an
ancilla $A$. A measurement is performed on the latter and depending on the
outcome $i$ different
unitaries can be applied to $S$ in order to extract work. The maximal amount of
extractable work using this protocol is the daemonic~ergotropy. }
\label{scheme}
\end{figure}

This work, which is dubbed {daemonic ergotropy}, is given by
\begin{equation}
\label{lab:daemerg}
W_D(\rho^{SA},H,M) = \tr[\rho^S H] - \sum_i \min_{U_i} \tr(\gamma^{S}_{i}\tilde H_i),
\end{equation} 
where $\tilde H_i=U_i\dg H U_i$, $M = \{\Pi_j\}$ is a projective measurement,
and $\gamma^S_i = \tr_A[\rho^{SA} (\id^S \ot \Pi^A_i)]$ is the unnormalised
conditional state of $S$ corresponding to the $i^\text{th}$ measurement outcome.
The daemonic ergotropy can be written in a more compact way using the ergotropy, namely 
\begin{equation}
W_D(\rho^{SA},H,M) = \sum_i W(\gamma^S_i,H).\label{lab:daem2}
\end{equation}

For a pure state, any
projective measurement $M$ with $\Pi_i$ rank-one projectors maximises the
daemonic ergotropy. In fact,
the conditional states $\gamma^{S}_{i}$ are then pure and it is always possible to find a unitary -- specific to every
conditional state -- that maps it to the ground state of the Hamiltonian, thus lowering as much as possible the energy of 
the system and extracting the maximum amount of work~\cite{daem}.

The difference between maximal daemonic ergotropy and ergotropy is called daemonic gain~\cite{daem}, and is formalised as
\mm{
\delta W(\rho^{SA},H) = \max_M W_D(\rho^{SA},H,M) - W(\rho^S,H).
}

If $\rho^{SA}$ is a pure product state, $\rho^S$ is pure. Thus, no measurement on the
ancilla is required for optimal work extraction, since in this case there is a unitary
that maps $\rho^S$ to the ground state of the Hamiltonian. Consequently, the
daemonic ergotropy coincides with the ergotropy in this case and there is no
daemonic gain.

The definitions provided above pinpoint the key role of the measurement step in such an ancilla-assisted extraction protocol. In particular, the assumption of 
projective measurements performed on $A$ appears to be too restrictive. It is
thus plausible to wonder if better performances of the daemonic work-extraction
scheme are possible when enlarging the range of possible measurements on the
ancilla to generalised quantum measurements. 

\section{Non-Optimality of Projective Measurements for Daemonic Ergotropy}
We now address such a scenario and provide an example where more energy can be
extracted from $S$ when generalised measurements are performed. 
To this end, we will employ the
formalism of positive
operator valued measures (POVMs)~\cite{heinosaari}. In the case of a finite 
 set of outcomes $\{i\}$, a POVM is a map that assigns a
positive semidefinite
operator $E_i$ -- dubbed as effect -- to each outcome $i$, such that $\sum_i E_i = \id$. As with
projective measurements, the probabilities for the outcomes are obtained as $p_i
= \tr(E_i \rho)$.
However, the effects $E_i$ of a POVM need not be projectors.

Let us consider now a three-level system $S$ and a two-level ancilla $A$ prepared in the joint state
\begin{align}
\rho^{SA} = \frac 13 \sum_{j=0}^{2} \ketbra jj \otimes \Pi\left(\frac{2\pi j}{3},0\right) \label{lab:povmbetprojst}
\end{align}
with projectors 
\begin{equation}
\Pi(\alpha,\beta) = \frac 12\{\id + \cos(\alpha) \sigma_z +
\sin(\alpha)[\cos(\beta)
\sigma_x - \sin(\beta) \sigma_y]\}.
\end{equation}

Here $(\alpha,\beta)$ are angles in the single-qubit Bloch sphere. 
We assume a reference Hamiltonian $H{=}\sum_j \epsilon_j \ketbra jj$ with energy eigenvalues $\epsilon_j$ arranged in increasing order. If only projective measurements $M$ are allowed on the state of the ancilla, the maximum daemonic ergotropy achieved upon optimizing over the measurement strategy is 
\begin{align}
\max_M W_D(\rho^{SA},H,M) = W(\rho^S,H) + \frac{\epsilon_2 - \epsilon_0}{2\sqrt 3}.
\end{align}

Details on this result are presented in Appendix~\ref{labpovmbetterproj}. However, if generalised measurements are permitted, one may choose the POVM with
effects $E_j = \frac 23 \Pi(2\pi j/3,0)$ to yield
a daemonic ergotropy~of 
\begin{align}
W_D(\rho^{SA},H,\{E_i\}) = W(\rho^S,H) + \frac 16(\epsilon_1 +
\epsilon_2-2\epsilon_0).
\end{align}

This can exceed the maximum daemonic ergotropy achieved through
projective measurements. For instance, we can assume to have shifted energy so
that $\epsilon_0 = 0$. Under such conditions, we would have
$W_D(\rho^{SA},H,\{E_i\})>\max_M W_D(\rho^{SA},H,M)$ for $ (\sqrt 3 - 1)
\epsilon_2 < \epsilon_1 \le \epsilon_2$.
Figure~\ref{fig-povmbetterproj} shows the daemonic gain $\delta W$ corresponding to the example above as a function of the value of the highest energy level of the 
Hamiltonian for projective measurements (PVMs) and POVMs. While in this example the optimal
projective measurement does not depend on the Hamiltonian, the optimal POVM
does. Therefore, the daemonic gain grows linearly with the value of the highest
energy value, as long as only projective measurements are taken into account.
For comparison, the daemonic gain that can be achieved with the previously
discussed POVM $\left(\frac 23 \Pi(2\pi j/3,0)\right)_j$ is plotted as a dashed line. 

\begin{figure}[H]
\centering
\includegraphics[width=.45\columnwidth]{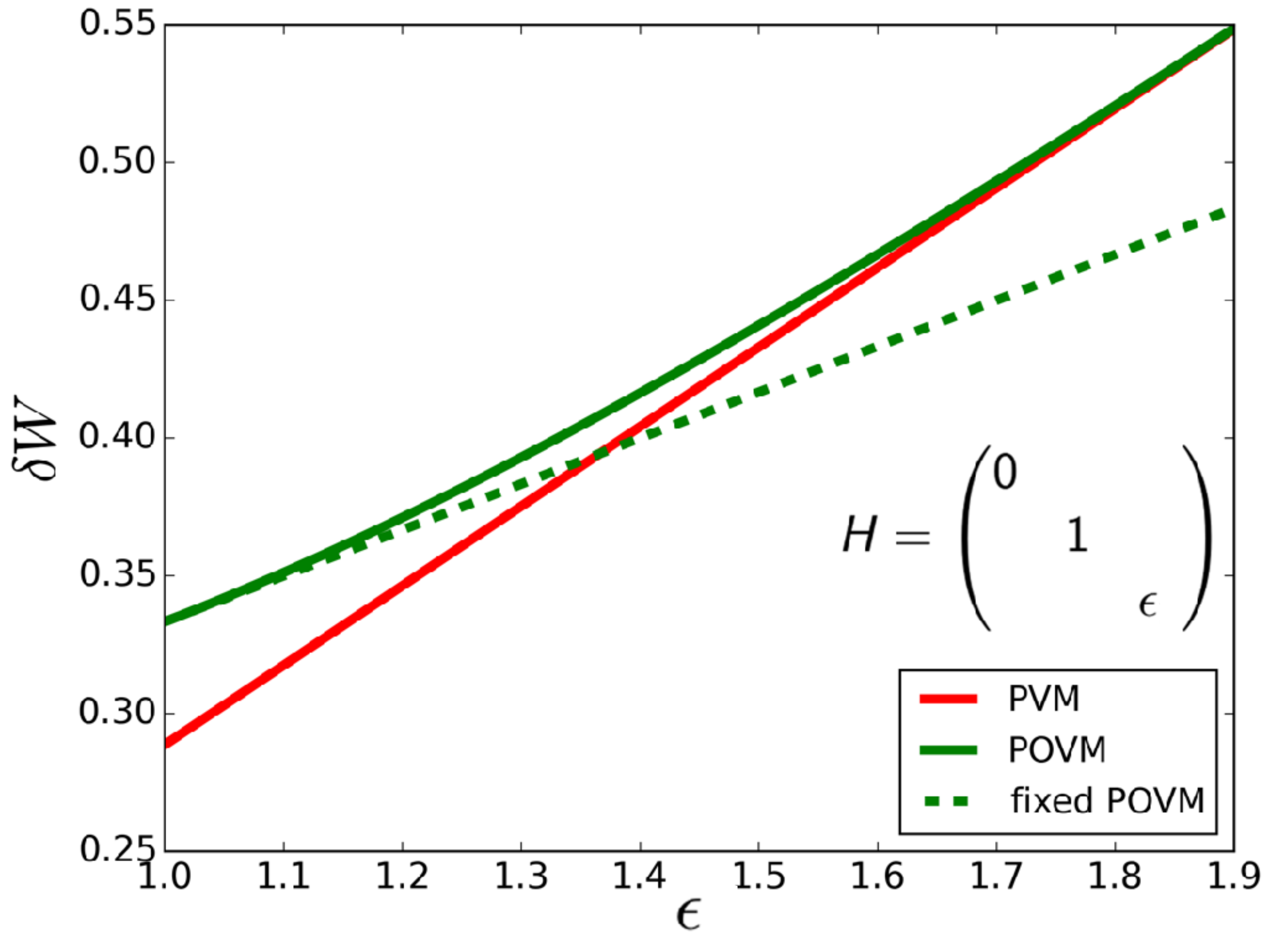}
\caption{Daemonic gain $\delta W$ as a function of the value of the highest
energy level of the Hamiltonian $H$ (in units of $\epsilon_1)$ for the state
$\rho^{SA}$ given in Equation~\eqref{lab:povmbetprojst}. Here $\epsilon=\epsilon_2/\epsilon_1$. We compare the performance under the optimal r projective measurements (PVM)
and  positive
operator valued measures (POVM). The latter was found numerically using  the see-saw algorithm proposed here. The former is 
determined analytically as discussed in Appendix~\ref{labpovmbetterproj}. The
dashed line is obtained as
the daemonic gain $\delta W$ for the fixed POVM with effects $ E_j = \frac 23
\Pi(2\pi j/3,0)$. }
\label{fig-povmbetterproj}
\end{figure}
 
 \section{Construction of Optimal POVMs}
 
 Having provided a useful example, we now move to address the problem of identifying the ideal POVM for optimal daemonic ergotropy. The following Lemma is instrumental to the achievement of our~goal:
\begin{lemma}
\label{lab:sublin}
The ergotropy is a sublinear function in its first argument, which refers to the
state. That is, for any $\gamma=\gamma_1+\gamma_2$
\mm{
    &W(\gamma,H) \le \sum_{i=1,2} W(\gamma_i,H) 
\intertext{and}
    &W(\lambda \gamma, H) = \lambda W(\gamma,H)
}
for any $\lambda \ge 0$. As ergotropy is symmetric under the
exchange of its first and the second argument, it is also sublinear in the
Hamiltonian.
\end{lemma}
\begin{proof}
The second equation holds trivially, which justifies our use of unnormalised
states. We obtain the first inequality as follows
\begin{align}
W(\gamma,H) &= \tr(\gamma  H) - \min_U\tr[U\gamma U\dg H]\nonumber \notag\\
&\le \sum_{j=1,2}\left[\tr(\gamma_j H) -  \min_U \tr(U \gamma_j U\dg H)\right]\notag\\\label{eq-ergsl1}
&=  \sum_{j=1,2}W(\gamma_j,H). 
\end{align}
\end{proof}
Note that sublinearity implies convexity, i.e., ${W[\lambda \gamma_1 +
(1-\lambda)\gamma_2,H] \le \lambda W(\gamma_1,H) + (1-\lambda)W(\gamma_2,H)}$.
This result allows us to state the following corollary:
\begin{corollary}
The daemonic ergotropy 
\mm{
 W_D(\rho^{SA},H,M) =& \sum_i W(\gamma^S_i,H) \ge  W(\sum_i \gamma^S_i,H) = W(\rho^S,H)}
is larger or equal to ergotropy. Equality holds for the trivial measurement,
with the identity as only effect.
\end{corollary}
This claim has already been proven in a
different way in Reference~\cite{daem}.
A second interesting consequence of the sublinearity of ergotropy is stated in the following
lemma:
\begin{lemma}
Daemonic ergotropy is a convex function of its third argument, which pertains
 to the measurement strategy.
\end{lemma}
\begin{proof}
Let us consider a mixed measurement strategy $Q=\lambda M + (1-\lambda)N$ with $0
\le \lambda\le 1$, and the corresponding daemonic ergotropy. We have
\begin{align}
    W_D[\rho^{SA},H,Q] &\le  \lambda\sum_i  W[\tr_A(\rho^{SA} \id \ot M_i),H]+ (1-\lambda) \sum_i W[\tr_A(\rho^{SA} \id \ot N_i),H]\notag\\
&{=} \lambda W_D(\rho^{SA},H,M){+} (1 {-} \lambda) W_D(\rho^{SA},H,N).
\end{align}
\end{proof}
We complete our formal analysis that precedes the presentation of an algorithm for the identification of the optimal POVM  with the following theorem.
\begin{Theorem}
\label{lab:thd2}
For any state $\rho^{SA}$ and any POVM $M$, one can find a POVM $\widetilde M$ with at
most $d^2$ effects, where $d$ is the dimension of the ancilla,
such that
\mm{
W_D(\rho^{SA},H,M) = W_D(\rho^{SA},H,\widetilde M).
}
\end{Theorem}
\begin{proof}
The set of POVMs on a $d$ dimensional system is convex and it has been shown
that the extremal points of this set are POVMs with at most $d^2$ effects~\cite{chiribella}. A convex function that is defined on a convex domain takes
its maximum on an extremal point. Therefore, there is an extremal POVM $E$ with
$n$ outcomes, $1 \le n \le d^2$, that
exhibits a daemonic ergotropy that is larger than or equal to the daemonic
ergotropy for $M$. If equality holds,
we choose $\widetilde M = E$. Otherwise, there is a mixture $\widetilde M =
\lambda E + (1-\lambda) I$
between $E$ and a trivial
random measurement $I$ with $n$ outcomes and effects $I_i = \id/n$ that
meets the requirement, since $W_D(\rho^{SA},H,I) = W(\rho^S,H) \le
W(\rho^{SA},H,M)$.
\end{proof}

We are now in the position to present an algorithm for the search of the optimal
measurement.
This task involves two parts (a) Finding the optimal
measurement and (b) Finding the optimal unitaries to calculate the ergotropies of
the conditional states. Assume a fixed measurement. Then, the conditional states
are fixed and one can find the optimal unitaries as discussed in the
introduction after Equation~(\ref{eq-1}). On the other hand, if some $d^2$ unitaries $U_i$ are given,
finding the optimal measurement $M = (E_i)_i$ is a semidefinite program (SDP)~\cite{roope}
\begin{align}
\min_{M}& \sum_i \tr(\tau_i E_i)\nonumber\\\notag
\text{s.t}\, &  
\sum_i E_i = \id\\ &
 E_i \ge 0 
\end{align}
where $E_i$ are the effects associated with the POVM $M$ and 
\mm{
\tau_i =
\tr_S(\rho^{SA} U_i\dg
H U_i).\label{lab:sigmai}}

We thus propose the following see-saw {Algorithm \ref{algorithm}}:
\begin{algorithm}[H]
\begin{spacing}{1.5}
\caption{Optimise POVM for daemonic ergotropy} 
\label{algorithm} 
		\begin{algorithmic}[1]%
		\State {Choose $n$ different unitaries $U_i$ and calculate $\tau_i$}
		\State {Solve the SDP above. This will yield a POVM $M$.}
		\State{Calculate the conditional states $\gamma^S_i$ for the POVM $M$ and the optimal unitaries $U_i$.}
	\Repeat  \Comment{Iterate steps 2 and 3}  
	\Until {convergence.}
\end{algorithmic}
\end{spacing} 
\end{algorithm}

We can restrict ourselves to $n = d^2$ different unitaries in the first step
because of Theorem \ref{lab:thd2}.
Calculating the daemonic ergotropy after every round of the algorithm will yield a
monotonically increasing sequence that is bounded from above because all involved
operators are bounded and will therefore converge. In the case of the example
discussed above, roughly 10 iterations are needed until the limit is
reached within numerical precision. The sequence however sometimes converges to
a local maximum that is strictly smaller than the
maximal daemonic ergotropy. Besides observing
this in practice, we also construct such a case in Appendix~\ref{labbadfixapp}.
 
 \section{The Role of Quantum Correlations}
 
 Notwithstanding the handiness of the algorithm built above, analytical solutions can be found in some physically relevant cases. The one most pertinent to the scopes of this work~\cite{daem} is embodied by quantum-classical $S$-$A$ states, i.e., states that can be cast in the form 
\mm{
\rho^{SA}_{qc} = \sum_j \sigma^S_j \ot \ketbra jj^A \label{eq-qc}
}
with $\{\ket j^A\}$ a set of orthonormal vectors and $\sigma^S_j$ unnormalised states.
This class of states has attracted attention from the community interested in
the characterization of general quantum correlations, for it has only classical
correlations, that is, it is not entangled and exhibits no quantum discord, if $A$ is
considered as the system the measurement being performed on~\cite{zurekdiscord,hendersonvedral,Modi,navascues}. 
For these states, we provide the following theorem. The proof is found in
Appendix~\ref{proofQC}.
\begin{Theorem}
\label{theoQC}
For a quantum-classical state
$\rho^{SA}_{qc}$, the maximum daemonic ergotropy is
\mm{
\max_M W_D(\rho^{SA},H,M) = \sum_j W(\sigma_j^S,H).
}
This value is achieved by performing the projective measurement with effects $P_j = \ketbra jj^A$ $(j =1,\ldots,d)$ on the ancilla $A$.
\end{Theorem}

This shows that, in the case of a quantum-classical state, we have an analytic form for the daemonic gain. 
To calculate it, we should diagonalise the reduced state $\rho^A = \tr_S({\rho^{SA}})$ of the ancilla. This  
yields a unitary
to make the state block-diagonal. The individual blocks are 
then the optimal conditional states $\sigma^S_j$ that one needs in order to compute the daemonic gain.

The above result paves the way
to an investigation on the role that quantum correlations play in the daemonic
protocol for work extraction. This important question was already partially
addressed in Reference~\cite{daem}, where a very close relation between daemonic gain
and entanglement in pure $S$-$A$ states was pointed out, while the link was shown to be looser for the case of mixed resource states. 

Here, by using the results reported above, we shed further light on the link between daemonic gain and quantum correlations. 
Let us assume that, for a given resource state $\rho^{SA}$, the optimal
measurement for daemonic gain is projective, and call $P_i=\ketbra ii$ the
corresponding projections, which can be chosen, without loss of generality, to
be rank one. We write the resource state as
\mm{
\rho^{SA} = \sum^S_{ij} \sigma^S_{ij} \ot \ketbra ij^A,
}
where the dyads $\ketbra ij^A$ are written in the basis defined by the optimal
projectors $P_i$ above. We notice that all off-block-diagonal terms $\sigma^S_{ij}$ (with $i\neq j$)
do not contribute to the daemonic gain, which is thus the same as the one associated with the
quantum-classical state 
\mm{
\rho^{SA}_{qc} = \sum_i \sigma^S_{ii} \ot \ketbra ii^A.
}

That this state is a quantum-classical state is obvious from the definition
provided in Equation~(\ref{eq-qc}).
This state can be produced by performing the optimal measurement and preparing
a pure state on the ancilla accordingly. This procedure destroys all the
quantum correlations, while the daemonic gain remains
unchanged. Quantum correlations in the resource states are
thus not useful, if the optimal measurement is projective.
This is especially true if only projective measurements are considered from
the start, which stresses the importance of considering generalised
measurements, if one aims at investigating the impact entanglement may have
on daemonic ergotropy.

However, we now show that, even if we allow for the use of arbitrary POVMs, the maximum
daemonic gain for any given Hamiltonian can be achieved by classical-classical states, i.e., states whose parties share only classical correlations~\cite{Modi}. 
We do this by providing an upper bound on the daemonic gain. This bound is tight as it is achieved by
a classical-classical state.  
Let us consider an explicit formula for
daemonic gain, where we have inserted the definitions of ergotropy and daemonic ergotropy. We have
\mm{
\delta W(\rho^{SA},H) &= \min_U \tr(U \rho^S U\dg H) - \min_{(E_k)}\min_{U_k}
\sum_k \tr(U_k \rho^S_k U_k\dg H).}

Using von Neumann's trace inequality, which reads $\tr{(AB)} \le \sum_i a_i b_i$
with $a_i (b_i)$ the eigenvalues of A (B) in increasing order, one easily finds that the first term never
exceeds $\frac 1{d_S} \tr(H)$, where $d_S$ is the dimension of the Hilbert space of $S$.
This value is attained if $\rho^S$ is maximally mixed.
The smallest value that the second term can take is
$\epsilon_0$, the lowest energy eigenvalue. This is achieved for pure
conditional states $\rho^S_k$. Consequently
\mm{
\delta W(\rho^{SA},H) \le \frac 1{d_S} \tr(H) - \epsilon_0.
}

If the dimension of the ancilla $d_A$ is greater or equal to $d_S$, this value is attained  by using -- among others -- the classical-classical state
\mm{
\rho^{SA} = \frac 1{d_S} \sum_{i=1}^{d_S} \ketbra{s_i}{s_i}^S \ot \ketbra{a_i}{a_i}^A
} 
 and the projective measurement with effects
$\ketbra{a_i}{a_i}^A$, where $\{\ket{a_i}^A\}$ ($\{\ket{s_i}^S\}$) forms an
orthogonal basis of $A$ ($S$). In the above example, the bound is also
achievable with maximally entangled pure~states
\mm{
\ket{\Psi^{SA}} = \frac 1{\sqrt{d_S}} \sum_{i=1}^{d_S} \ket{s_i}^S\ket{a_i}^A.
}

The maximal daemonic gain is, however, not always achieved using pure states,
as the following example shows.
Consider the following classical-classical state with a qutrit system and a qubit ancilla
\begin{equation}
\rho^{SA} = \frac 13[ \ketbra 00^S\ot\ketbra 00^A + (\ketbra 11^S + \ketbra
22^S) \ot \ketbra 11^A]. \label{lab:eqcc}
\end{equation}

For a Hamiltonian with eigenvalues $\epsilon_0 \le \epsilon_1 \le \epsilon_2$
one easily finds the daemonic gain $\delta W(\rho) = (\epsilon_2 - \epsilon_0)/3$.
On the other hand, for any pure
state, including maximally entangled states, we have
\mm{
\delta W(\ket{\Psi}^{SA}) \le \frac 12(\epsilon_1 - \epsilon_0),}
since the Schmidt-rank of a pure state on a $3\times 2$ dimensional system is at
most $2$.
For a suitably chosen
Hamiltonian, such as $H/\epsilon_1 = \ketbra 11 + \epsilon \ketbra 22$, with
$\epsilon=\epsilon_2/\epsilon_1 > 3/2$, the daemonic gain of
$\rho^{SA}$ [Equation~(\ref{lab:eqcc})] exceeds the daemonic gain of any pure state of the
same system.

\section{Multipartite Daemonic Ergotropy}

In this section we want to investigate a multipartite adaptation of the daemonic ergotropy protocol. Concretely, we consider the situation in which
$N$ different parties $i \in \{1, ..., N\}$ each own one system $S_i$, whose energy they can locally measure using their local Hamiltonian $H^{(i)}$. The energy of all systems combined will then be evaluated using the Hamiltonian

\begin{align}
H = \sum_{i=1}^N H^{(i)}.
\end{align}

Additionally, they can only act on their systems locally, that is using local unitaries. It is only this restriction that makes the protocol multipartite regarding the systems. If arbitrary global unitaries were admitted, this would be equivalent to a situation with a single system consisting of $N$ subsystems.

We also take the case into account in which there are $M$ ancillas, each owned by a different party $k \in \{1,...,M\}$. As we are interested in a genuinely multipartite protocol, each party must resort to local measurements, possibly assisted by classical communication among the parties, yielding outcomes $j_k$. After all outcomes are obtained, they are publicly announced and every party $i$ performs a unitary on their system $S_i$, which may depend on all the outcomes $\vec j = (j_k)_{k=1}^M$. We define the multipartite daemonic ergotropy $W_D^{\text{mult}}$ to be the maximum amount of energy that can be extracted from a state in this way.

Note that, in spite of the previously imposed restrictions, our notion of multipartite daemonic ergotropy is in fact a generalisation of daemonic ergotropy. This might appear paradoxical at first glance. However, the daemonic ergotropy protocol is equivalent to the protocol of multipartite daemonic ergotropy for one system and one ancilla. This especially includes scenarios in which system and ancilla comprise several subsystems. Studying multipartite daemonic ergotropy is interesting, because it is also applicable to settings, in which the implementation of global measurements and unitaries are unfeasible.


As we are only concerned with local measurements, possibly assisted by classical
communication among the parties, all effects of a POVM are
of the form
\mm{
E_{\vec j} = \bigotimes_{k=1}^M E^k_{j_k}.
}

We denote the respective conditional states of all systems by $\rho^S_{\vec j}
= \tr_{(A_1 ... A_M)}{(\rho^{S_1...S_N A_1 ... A_M} E_{\vec j})}$ and the conditional
state of system $S_i$ given a measurement outcome $\vec j$ as $\rho^i_{\vec j}$.
As before, the multipartite daemonic ergotropy can then be
expressed in terms of the ergotropy as
\mm{
W^{\text{mult}}_D(\rho^{\{S_j\},\{A_k\}},H,E) = \sum_{\vec j} \sum_{i=1}^N W(\rho^i_{\vec
j}, H^{(i)}).}

With this result, we can show that contrary to the bipartite case [cf.
discussions after Equation~(\ref{lab:daem2})] in the multipartite setting projective measurements are in general not
optimal for work extraction even for pure states.
In order to see this, consider a state $\rho^{S_1 A}$ and a purification $\ket{\psi}^{S_1
S_2 A}$, with $\rho^{S_1 A} = \tr_{S_2}(\ketbra{\psi}{\psi}^{S_1 S_2 A})$. If we now
assume that system $S_2$ is equipped with a local Hamiltonian $H^{(2)} = h \id$,
where $h$ is a constant, the multipartite daemonic ergotropy of the purified
state is
\begin{equation}
\begin{aligned}
W^{\text{mult}}_D(\ket{\psi}^{S_1 S_2 A},H,E) &=\sum_{\vec j} \left[W(\rho^1_{\vec
j},H^{(1)}) + W(\rho^{2}_{\vec j},H^{(2)})\right]\\
& = \sum_{\vec j} W(\rho^1_{\vec j},H^{(1)})\\
& = W_D(\rho^{S_1A},H^{(1)},E).
\end{aligned}
\end{equation}

This result stems from the fact that $H^{(2)}$ is completely
degenerate and the ergotropy vanishes for such Hamiltonians.
Thus, also the multipartite daemonic ergotropy of the purification
is maximised for the same POVM that also maximises the daemonic ergotropy of
$\rho^{SA}$. Hence, the purification of the qutrit-qubit state stated in
Equation~(\ref{lab:povmbetprojst}) is an
example for a pure state that requires a POVM to maximise the multipartite daemonic
ergotropy.
Note, however, that there are also states for which projective measurements are optimal
independently of the choice of the Hamiltonian. The first example are states
that possess a Schmidt decomposition \cite{peres}, i.e.,
\mm{
\ket{\Psi} = \sum_i \sqrt{\lambda_i} \ket{i_{S_1} \ldots i_{S_n} i_{A_1} \ldots
i_{A_m}},} with $\braket{i_{S_l}}{j_{S_l}} = \braket{i_{A_l}}{j_{A_l}} =
\delta_{ij} \forall i,j,l$.
For qubits, these are exactly the states that become separable as
soon as one particle is ignored \cite{neven}. A famous example is the $m$-partite
Greenberger--Horne--Zeilinger (GHZ)~state
\mm{
\ket{GHZ} = \frac 1{\sqrt 2}(&\ket{0_{S_1}\ldots 0_{S_n} 0_{A_1} \ldots
0_{A_m}}+ \ket{1_{S_1}\ldots 1_{S_n} 1_{A_1} \ldots 1_{A_m}}),
}
for which the local projective measurements on $\ket 0$ and $\ket 1$ are optimal,
since the conditional state of all systems is a pure product state independently
of the outcome and its
energy can thus be minimised using local unitaries.

A second class of states for which projective measurements are always optimal
are multipartite quantum-classical states
\mm{
\rho_{S_1 \ldots S_n A} = \sum_i \sigma^{S_1 \ldots S_n}_i \ot \ketbra ii^A.
}
Here, we can recover the proof of Theorem 5 to show that the
projective measurement with projectors $\ketbra ii$ is  optimal. The only adaptation to the proof
is that the unitaries are now required to be products. Of course this result is
still true in the special case when the ancilla is made up of several parties,
such that the state can be written as
\mm{
\rho^{\{S_j\} \ldots \{A_m\}} = \sum_i \sigma_i^{S_1 \ldots S_n} \ot
\ketbra ii^{A_1} \ot \ldots \ketbra ii^{A_m}.
}

In this case, the optimal measurement 
consists of the local projective measurements
with effects $\ketbra ii_{A_k}$.
 
\section{Conclusions}

We have significantly extended the concept of daemonic
ergotropy to situations involving POVM-based information-gain processes,
demonstrating that, in general, one should expect an advantage coming from the
use of generalised quantum measurements in ancilla-assisted work-extraction
schemes. While the optimal generalised measurements can be identified
analytically in some restricted -- yet physically relevant -- cases, we have
proposed an SDP-based see-saw algorithm for their construction.  This has led to
a number of results shedding light on previously unreported issues linked to
daemonic approaches to quantum work extraction: while the interplay between
quantum correlations and the features of the optimal measurements appears to be
intricate, the structure of entanglement sharing in a multipartite scenario where only local unitaries and POVMs are used turns out to be key in the performance of ancilla-assisted work extraction. 

Our work paves the way to a number of interesting developments aimed at
exploring further and clarifying the relation between quantum features and
work-extraction games in quantum scenarios. On the one hand, it will be very
interesting to further compare, quantitatively, the performance of daemonic protocols under optimal PVMs and POVMs to ascertain the extents of the benefits induced by the latter class of measurements against the difficulty of practically implement them. On the other hand, the analysis that we have reported here leaves room to the in-depth assessment of multipartite daemonic gain against the structure of multipartite entanglement aimed at the identification of potentially {\it optimal} classes of multipartite entangled states, when gauged against their role as a resource in work-extraction schemes. 

\vspace{6pt} 




\funding{M.P. acknowledges support by the SFI-DfE Investigator
Programme (grant 15/IA/2864), the H2020 Collaborative Project TEQ (Grant
Agreement 766900), the Leverhulme Trust Research Project Grant UltraQuTe (grant
nr. RGP-2018-266) and the Royal Society Wolfson Fellowship (RSWF\textbackslash
R3\textbackslash183013). O.G. acknowledges support by the DFG and the ERC
(Consolidator Grant 683107/TempQ).}

\acknowledgments{F.B. thanks the Centre for Theoretical Atomic, Molecular
and Optical Physics for hospitality while developing on part of this
work as well as the Studienstiftung des Deutschen Volkes e.V. and the House of
Young Talents Siegen.}


\appendixtitles{yes} 
\appendix
\appendixsections{multiple}
\section{POVM Advantage in Qutrit-Qubit Example}
 \label{labpovmbetterproj}
 We present the state
\begin{align}
\rho^{SA} = \frac 13 \sum_{j=0}^{2} \ketbra jj \otimes P_j \\ 
\intertext{with}
P_j = \Pi\left(\frac{2\pi j}{3},0\right) \label{eq-povm-pi}
\end{align}
and
\begin{equation}
\Pi(\alpha,\beta) = \frac 12\{\id + \cos(\alpha) \sigma_z +
\sin(\alpha)[\cos(\beta)
\sigma_x - \sin(\beta) \sigma_y]\}
\end{equation}
as an example in which higher daemonic ergotropy can be achieved with POVMs
compared to projective measurements, if a Hamiltonian is chosen suitably.
Here, we work out the details and show all necessary calculations explicitely.
First, we find the optimal projective measurements. It turns out, that they can
be found independently of the chosen Hamiltonian. With this result and bearing
in mind that the daemonic gain is invariant under unitary transformations of the
Hamiltonian, we can then compute the daemonic ergotropy as a function of
the energy spectrum.

Since the ancilla is a qubit, there are only two types of projective
measurements: Either, the projective measurement has one outcome that is
obtained with certainty, which makes the measurement trivial, or the measurement
has two outcomes. In the latter case, the effects are rank one. Therefore, we can compute the maximal daemonic
gain for projective measurements by computing it for the measurement $\mathbf
\Pi = \lb\Pi(\alpha,\beta), \Pi(\alpha + \pi, \beta)\rb$
and optimise over the angles $\alpha$ and $\beta$ afterwards. We have
\mm{
   \rho^S = &\frac 13(\ketbra 00 + \ketbra 11 + \ketbra 22), \notag\\
   \rho^S_\alpha = &\tr[\rho^{SA}(\id \otimes \Pi(\alpha,\beta))] \notag\\
   = & \frac 13 \{\ketbra 00 \tr[P_0 \Pi(\alpha,\beta)] + \ketbra 11 \tr[P_1 \Pi(\alpha,\beta)]+ \ketbra 22 \tr[P_2 \Pi(\alpha,\beta)]\}\nonumber\\
   = & \frac 13\left[\ketbra 00 \frac 12(1 + \cos(\alpha)) \right.+ \ketbra 11 \left(\frac 12 - \frac 14 \cos(\alpha) + \frac{\sqrt 3}4 \sin(\alpha) \cos(\beta)\right)\nonumber\\
  +& \ketbra 22 \left(\frac 12 - \frac 14 \cos(\alpha) - \frac{\sqrt 3}4
\sin(\alpha) \cos(\beta)\right)\Bigg], \notag}
\mm{
    \rho^S_{\alpha + \pi}
   = &\tr[\rho_{SA}(\id \otimes \Pi(\alpha + \pi,\beta))]\notag \\
     = & \frac 13\left[\ketbra 00 \frac 12(1 - \cos(\alpha)) \right.+ \ketbra 11 \left(\frac 12 + \frac 14 \cos(\alpha) - \frac{\sqrt 3}4 \sin(\alpha) \cos(\beta)\right)\nonumber\\
  +& \ketbra 22 \left(\frac 12 + \frac 14 \cos(\alpha) + \frac{\sqrt 3}4
\sin(\alpha) \cos(\beta)\right)\Bigg] . 
   }

From the definition of ergotropy one can easily see that the ergotropy of the
conditional states $\gamma^S_{\alpha}$
and $\gamma^S_{\alpha+\pi}$ will be maximal for $\cos(\beta) = \pm 1$. This becomes
clear when considering a state
\mm{
\rho = a \ketbra 00 + (b+c) \ketbra 11 +  (b-c) \ketbra 22,
}
where $a,b,c \in \mathbb R$ and $c\ge 0$. Let the Hamiltonian be
\mm{
H = \epsilon_0\ketbra {\epsilon_0}{\epsilon_0}  + \epsilon_1 \ketbra{\epsilon_1}{\epsilon_1} + \epsilon_2 \ketbra{\epsilon_2}{\epsilon_2}.
}

Then, the ergotropy can without loss of generality be written as
\mm{
W &= \tr[\rho H] - \min_U \tr[U\rho U\dg H] \notag\\
 &= \tr[\rho H] - [\epsilon_0 a + \epsilon_1(b+c) + \epsilon_2(b-c)]\notag\\
  &= \tr[\rho H] - [a \epsilon_0 + b(\epsilon_1 + \epsilon_2) + c(\epsilon_1 - \epsilon_2)],
}
where the energy eigenvalues are ordered such that the minimum is
achieved. Consequently, we get $\epsilon_1\le \epsilon_2$ since $(b+c) \ge (b-c)$. Therefore, $W$ increases with $c$ and we can set $\beta = 0$ in the above calculation.
Exploiting addition theorems, we can now write
\mm{
    \rho^S_{\alpha}
    =& \frac 16\left[ \ketbra 00 [1 + \cos(\alpha)] + \ketbra 11 \lb 1 +
\cos\lb\alpha - \frac{2\pi}3\rb \rb\nonumber \right. + \ketbra 22 \lb1 +
\cos\lb\alpha + \frac{2\pi}3\rb\rb\bigg]  \notag\\
    \rho^S_{\alpha+\pi}
    =& \frac 16\left[ \ketbra 00 (1 - \cos \alpha) + \ketbra 11 \lb 1 -
\cos\lb\alpha - \frac{2\pi}3\rb \rb \nonumber \right. + \ketbra 22 \lb1 - \cos\lb\alpha +
\frac{2\pi}3\rb\rb\bigg]. \label{labcond2eq}
    }

As one can easily see, an optimal value of $\alpha$ is not unique, as shifting its value by $\frac{2\pi}3$ can be compensated
by relabeling the states, which does not affect the daemonic gain. We now aim to find the optimal $\alpha$ in the interval
$[-\frac{\pi}3,\frac{\pi}3)$. When calculating the ergotropy of the conditional states we need to know the ordering of their
eigenvalues
\mm{
\alpha \in \left[-\frac{\pi}3,0\rb \Rightarrow \cos(\alpha) \ge \cos\lb \alpha +
\frac{2\pi}3\rb \ge \cos\lb \alpha - \frac{2\pi}3\rb \notag\\
\alpha \in \lb0,\frac{\pi}3\rb \Rightarrow \cos(\alpha) \ge \cos\lb \alpha - \frac{2\pi}3\rb \ge \cos\lb \alpha + \frac{2\pi}3\rb
}

In the following calculation, the upper sign will refer to the negative and the lower sign will refer to the positive interval
\mm{
\delta W(\rho_{SA},H,\mathbf \Pi) & = W_D(\rho_{SA},H,\mathbf \Pi) - W(\rho_S,H)
\notag \\
&= \tr[\rho_S H] - \min_{\mathbf \Pi} \sum_k \tr[\rho_SA (U_k\dg H U_k \otimes
\Pi_k)]- \left[\tr[\rho_S H] - \min_U \tr[\rho_S U\dg H U]\right]\notag \\
& = \min_U \tr[\rho_S U\dg H U] - \min_{\mathbf \Pi} \sum_k \tr[\rho_SA (U_k\dg H
U_k\otimes \Pi_k)] \notag \\
& = \max_{\alpha}\bigg\{ \frac 13(\epsilon_0 + \epsilon_1 + \epsilon_2) - \frac
16( \epsilon_0 [1 + \cos (\alpha)]) + \epsilon_1 \lb1 + \cos\lb\alpha \pm \frac{2\pi}3\rb\rb \nonumber\\
  &\quad + \epsilon_2 \lb 1 + \cos\lb\alpha \mp \frac{2\pi}3\rb \rb+ \epsilon_0 \lb 1 - \cos\lb\alpha \mp \frac{2\pi}3\rb \rb\nonumber\\
  &\quad + \epsilon_1 \lb1 - \cos\lb\alpha \pm \frac{2\pi}3\rb\rb+ \epsilon_2 (1 - \cos \alpha)\bigg\}\notag \\
&= \frac 16(\epsilon_2 - \epsilon_0) \max_{\alpha}\lb\cos(\alpha) - \cos\lb\alpha
\mp \frac{2\pi}3\rb\rb \nonumber \\ &= \frac{\epsilon_2 - \epsilon_0}{2\sqrt 3}.
}

Now, that we computed the maximal daemonic gain for projective measurements, we compare this with the daemonic gain that can be
achieved by using the POVM $M$, consisting of the effects $\frac 23 P_i$, as
defined in Equation~(\ref{eq-povm-pi}). In this case, the conditional states are
\mm{
\gamma^S_{P_0} = \frac 29 \lb \ketbra 00 + \frac 14 \ketbra 11 + \frac
14 \ketbra 22 \rb, \notag \\
\gamma^S_{P_1} = \frac 29 \lb \frac 14 \ketbra 00 +  \ketbra 11 + \frac
14 \ketbra 22 \rb, \notag \\
\gamma^S_{P_2} = \frac 29 \lb \frac 14 \ketbra 00 + \frac 14 \ketbra 11 +
\ketbra 22 \rb.
}

Given the conditional states, we can now compute the daemonic gain as
\begin{equation}
\begin{aligned}
\delta W =& \epsilon_0 \left(\frac 13 - \frac 23\right) + \epsilon_1 \left(\frac 13 - \frac 16\right) + \epsilon_2 \left(\frac 13 - \frac 16\right) \\
=& -\frac 13 \epsilon_0 + \frac 16 (\epsilon_1 + \epsilon_2).
\end{aligned}
\end{equation}

Choosing the Hamiltonian $H = \ketbra{\epsilon_1}{\epsilon_1} + \ketbra{\epsilon_2}{\epsilon_2}$ provides an example where
the maximal daemonic gain can not be achieved by using projective measurements because
\mm{
\delta W_{\text{proj}} = \frac 1{2\sqrt3} < \delta W_{M} = \frac 13.}
   \section{\label{labbadfixapp} Non-Optimal Convergence of the See-Saw Algorithm}
In the following, we construct a case in which Algorithm \ref{algorithm} will yield a sequence
 of values for the daemonic ergotropy that does not
converge against the maximal daemonic ergotropy. 
\label{lab:lfix}
Consider a state $\rho^{SA}$ on a system $S$ with a Hamiltonian $H$ and a
$d$-dimensional ancilla $A$, such that the optimal measurements are rank-one
projective measurements as long as only $d$-outcome measurements are considered. Then, there exists an
initialisation of Algorithm \ref{algorithm}, such that the
sequence of daemonic ergotropies generated by the algorithm limits in the
maximal daemonic ergotropy for $d$-outcome measurements.
   In order to see this, consider a measurement $\Pi$ that is optimal among $d$-outcome measurements. For
the effects $\{\Pi_1,\ldots, \Pi_d\}$ one finds $d$ optimal unitaries $\{V_1, \ldots,
V_{d}\}$. We now
initialise the algorithm for $d^2$ outcomes in the following~way
\mm{
U_i &= V_i, \,\,\,\, \quad i = 1, ..., d-1\nonumber\\
U_i &= V_{d}, \quad i = d, ..., d^2. \label{eq-b1}
}

This implies $\tau_{d} = \tau_{d + 1} = \ldots = \tau_{d^2}$, where
$\tau_i~=~\tr_S(\rho^{SA} U_i\dg H U_i)$. Hence, the
objective of
step~2 of the algorithm simplifies to
\mm{
&\min_M \sum_{i=1}^{d^2} \tr(\tau_i M_i) 
= \min_M \left[\sum_{i=1}^{d-1} \tr(\tau_i M_i){+}\tr\left(\tau_{d}
\sum_{j=d}^{d^2} M_j\right)\right].
}

The value of this expression thus depends on $d$ effects $M_1,\ldots,M_{d-1},
\sum_{j=d}^{d^2} M_j$. In this case, the minimum can by assumption only be achieved
if the effects are all rank-one. This implies that the first $d - 1$ effects are
orthogonal rank-one projectors and the remaining effects are rank-one operators
on the remaining one-dimensional subspace and sum up to a rank-one projector.
Thus, the algorithm again finds a $d$-outcome rank-one projective measurement
that is optimal among $d$-outcome measurements.
The case that was discussed above is of practical relevance, as we have observed
in numerical experiments
that randomly initialised unitaries may converge against the configuration
stated in Equation~(\ref{eq-b1}).

The example discussed in Appendix~\ref{labpovmbetterproj} meets the requirement
that all optimal
two-outcome measurements are rank-one projective measurements. The optimal projective
measurements are calculated in Appendix~\ref{labpovmbetterproj}. Any two outcome
measurement in two dimensions
with rank-two effects can be considered as a mixture of a rank-one projective
measurement with white noise. The only case, in which white noise will not
decrease the daemonic ergotropy is, if the conditional
states $\gamma^S_i$ [Equation~(\ref{labcond2eq})] are simultaneously
 diagonalisable by the same diagonalising unitary and with the same ordering of
eigenvalues in diagonal form. This is however not the case, since both states
are already diagonal but the eigenvalues are not in the same order.

In the same example, the maximum daemonic ergotropy cannot be achieved with
$d$-outcome measurements.

\section{Proof of Theorem~\ref{theoQC}}
\label{proofQC}

In this Appendix we provide a complete proof of the statement made in
Theorem~\ref{theoQC}, which we repeat here again for easiness of reading. 
For a quantum-classical state, that is a state that can be cast in the form\mm{
\rho^{SA}_{qc} = \sum_j \sigma^S_j \ot \ketbra jj^A
}
with a set of orthonormal vectors $\{\ket j^A\}$  and unnormalised states
$\sigma^S_j$ the following theorem holds. 

\noindent
{\bf Theorem 5.} \emph{For a quantum-classical state
$\rho^{SA}_{qc}$, the maximum daemonic ergotropy is
\mm{
\max_M W_D(\rho^{SA},H,M) = \sum_j W(\sigma_j^S,H).
}
This value is achieved by performing the projective measurements $P_j = \ketbra
jj^A$ on the ancilla $A$.}
\begin{proof}
The first claim follows directly from the second claim using Equation~(\ref{lab:daem2}).
Therefore, we prove the second claim by showing that the daemonic gain achieved through  any POVM $E$ with effects
$E_i$ and an arbitrary number of outcomes $N$ has an upper bound given by the value corresponding to the use of projective measurements. 
We start by computing the conditional states
\begin{equation}
\begin{aligned}
\gamma^S_k = \tr_A\left[{\rho^{SA} (\id \ot E_k)}\right]=  \sum_{j=1}^d \sigma^S_j \bra j E_k \ket j.
\end{aligned}
\end{equation}

It can be easily seen that post-processing can never increase the daemonic
ergotropy. This allows us to assume, without loss of generality, that all effects
are rank-one and use Naimark's extension theorem~\cite{Nielsen} to write
\begin{equation}
\gamma^S_k = \sum_{j=1}^N \sigma^S_j |\braket j{\phi_k}|^2,
\end{equation}
{where $(\ketbra{\phi_k}{\phi_k})_{k=1}^N$ is the Naimark extension 
of the operators $E_k$ on the extended ancilla space. Then,
$(\ket{\phi_k}))_{k=1}^N$ is an orthonormal basis in the extended ancilla space. We also extend 
$(\ket j)_{j=1}^d$, so $(\ket j)_{j=1}^N$ is another orthonormal basis in the
extended ancilla space and set  $\sigma^S_j = 0,~\forall j>d$. We can now interpret $|\braket j{\phi_k}|^2$ as entries of a
doubly stochastic matrix and apply the Birkhoff-von Neumann theorem~\cite{Lenard1978},
which allows us to express this doubly stochastic matrix as a convex combination
of permutation matrices $\pi^{(n)}=\left(\pi^{(n)}_{jk}\right)_{jk}$. This yields}
\begin{equation}
\gamma^S_k =\sum_{j=1}^{N} \sigma_j \sum_n p_n \pi^{(n)}_{jk}
\end{equation}
with probabilities $p_n$.

We insert this result into the formula of the daemonic ergotropy 
\mm{
W_D(\rho^{SA},H,M) = \tr(\rho^S H) - \sum_k \min_{U_k} \tr(U_k
\gamma^{S}_{k}U_k\dg  H).
}

As we are interested in the optimal measurement, our only
concern is the second term   
\begin{align}
&\sum_{k=1}^N \min_{U} \tr(U  \gamma^S_k U\dg H)\notag\\
= &\sum_{k=1}^N \min_{U} \tr(U \sum_{j=1}^N \sigma_j \sum_n p_n \pi^{(n)}_{jk} U\dg
H) \notag\\
\ge &\sum_{k,j,n} p_n \pi_{jk}^{(n)} \min_{U} \tr(U \sigma_j
U\dg H) \notag\\
=& \sum_n p_n  \sum_j \min_{U}\tr(U \sigma_j U\dg
H)\sum_k \pi_{jk}^{(n)}\notag\\
=&\sum_j \min_{U}\tr(U \sigma_j U\dg H),
\end{align}
which is bounded from below by the value that
is achieved for the projective measurement $P_j = \ketbra jj$, as stated above.
\end{proof}

\reftitle{References}




\end{document}